\documentclass[11pt]{article}
\usepackage{amsmath,amsthm,amssymb}
\date{}
\newtheorem{theorem}{Theorem}
\newtheorem{lemma}{Lemma}

\newtheorem{proposition}{Proposition}

\newtheorem{open problem}{Open problem}

\title{The Maximum Nonlinearity of 8-variable Balanced Boolean Functions is 116}
\author{ Yueying~Lou\footnote{School of Computer and Electronic Information/School of Artificial Intelligence,Nanjing Normal University, Nanjing, China. Email: 232212001@njnu.edu.cn.
} ,\ Youle~Xu\footnote{China Telecom Corporation Limited, Huaibei Branch, China. Email: youlexu@163.com.} \ and Qichun~Wang \footnote{School of Computer and Electronic Information/School of Artificial Intelligence,
Nanjing Normal University, Nanjing, China.
 Email: qcwang@fudan.edu.cn.} \footnote{Corresponding author.}}
\begin{document}
\maketitle
\begin{abstract}
Does there exist an 8-variable balanced Boolean function with nonlinearity 118? This is a famous and long-standing open problem in Boolean functions. In this paper, we give a negative answer to this open problem, and prove that the maximum nonlinearity of balanced Boolean functions in 8 variables
is 116, establishing the 8-variable instance of Dobbertin's conjectured recursion for the minimum Walsh peak of balanced functions. Equivalently, the covering radius of the Reed--Muller code RM(1,8) in 0-resilient Boolean functions is 116.
We deduce the result by obtaining a quadratic form of polar rank 6, and reducing the problem to the existence of a 6-variable Boolean function with nonlinearity greater than 26 whose Hamming distance to a bent function is between 4 and 6.

\end{abstract}



\section{Introduction}

Boolean functions are fundamental objects in cryptography and coding theory \cite{Car,Logachev}. An $n$-variable Boolean function is a map from the $n$-dimensional vector space $\mathbb{F}_{2}^{n}$ to the binary field $\mathbb{F}_{2}$.
To be used in the design of symmetric ciphers, Boolean functions should be balanced and with high nonlinearity \cite{Cusick}. A $2n$-variable bent function has the maximum nonlinearity $2^{2n-1}-2^{n-1}$. However, it is not balanced.
The upper bound for the nonlinearity of $2n$-variable balanced Boolean functions is $2^{2n-1}-2^{n-1}-2$, which is tight for $n\le 3$. It is an open problem whether this bound is always tight. In particular, does there exist an 8-variable balanced Boolean function with nonlinearity 118? This is a famous and long-standing open problem in Boolean functions (see e.g., \cite{Car,Car0,Djurasevic,Maitra,Picek}). 

Let $NLB_n$ be the maximum nonlinearity of $n$-variable balanced Boolean functions.  It has been conjectured by Dobbertin \cite{Dobbertin} that
\[
NLB_{2n}=2^{2n-1}-2^n+NLB_n,
 \]
 which is still unsettled for $n\ge 4$ \cite{Oz}. Particularly, if the conjecture  is true, then the minimum Walsh peak of 8-variable balanced Boolean functions should be 24, and $NLB_8=116$.

 The $r$-th order Reed--Muller code of length $2^m$, denoted by RM$(r,m)$, is the binary linear code formed
by truth tables of all $m$-variable Boolean functions of algebraic degree $\leq r$.  The covering radius of RM$(1,m)$ is the maximum nonlinearity of $m$-variable Boolean functions, which is known for $m$ even or $m<9$. There are 9-variable unbalanced Boolean functions with nonlinearity 242 in the generalized rotation symmetric class \cite{Kavut,Kavut2}, thus the covering radius of RM(1,9) is at least 242. For odd $m>9$,  the covering radius of RM$(1,m)$ is completely open.

 To resist fast correlation attacks and linear attacks, Boolean functions used in stream ciphers should be $t$-resilient and with high nonlinearity. So, covering radius of Reed--Muller codes for $t$-resilient functions was defined \cite{Borissov,Iwata}, and the covering radius of RM(1,8) in 0-resilient Boolean functions is the same as $NLB_8$.

In this paper, we give a negative answer to that famous open problem, and  prove that $NLB_8=116$, establishing the 8-variable instance of Dobbertin's conjectured recursion for the minimum Walsh peak of balanced functions. Equivalently, the covering radius of the Reed--Muller code RM(1,8) in 0-resilient Boolean functions is 116. It should be noted that a very recent preprint work \cite{Yuan} claims a computer-assisted proof of the same result, using 7-variable restrictions and Farkas certificates. Our proof is obtained independently and it is computer-free.

 The rest of the paper is organized as follows. In Section~2, the necessary background is established. In Section~3, we deduce the result by obtaining a quadratic form of polar rank 6, and reducing the problem to the existence of a 6-variable Boolean function with nonlinearity greater than 26 whose Hamming distance to a bent function is between 4 and 6. We end in Section~4 with conclusion and open problems.

\section{Preliminaries}
Let $\mathbb{F}_{2}$ denote the binary finite field, and let $\mathbb{F}_{2}^{n}$ be the $n$-dimensional vector space over $\mathbb{F}_{2}$. A {\em Boolean function} of $n$ variables is a map from $\mathbb{F}_{2}^{n}$ to $\mathbb{F}_{2}$. We denote by $B_{n}$ the set
of all $n$-variable Boolean functions.

Every Boolean function has a unique {\em algebraic normal form} (ANF),
which is its multivariate polynomial representation over $\mathbb{F}_{2}$:
\[
f(x_1,\ldots ,x_n)=\oplus_{S\subseteq \{1,2,\ldots ,n\}}a_S\prod_{i\in S}x_i.
\]
 The {\em algebraic degree} of $f$, denoted by $\deg(f)$, equals $\max\{|S| :  a_S\neq 0 \}$. The function $f$ is {\em affine}, if $\deg(f)\le 1$. We use $A_n$ to denote the set of all $n$-variable affine functions. That is, $A_n=\{g\in B_n: \deg(g)\le 1\}$.

 Let $f\in B_n$ and
\[
supp(f)=\{x\in \mathbb{F}_{2}^{n}: f(x)=1\}.
  \]
 The cardinality $|supp(f)|$ is called the {\em Hamming weight} of $f$, and will be denoted by $wt(f)$. The function $f$ is  {\em balanced} if $wt(f)=2^{n-1}$.

Let $f, g\in B_n$. The {\em Hamming distance} between $f$ and $g$ is $d(f,g)=wt(f\oplus g)$. The {\em nonlinearity} of $f$, denoted by $nl(f)$, is the minimum distance between $f$ and the set of affine functions. That is, $nl(f)=\min_{g\in A_n} d(f,g)$.

The Reed--Muller code RM$(r,m)$ is the set of truth-tables of all $m$-variable Boolean functions whose
algebraic degree is at most $r$. The {\em covering radius} of $RM(r,n)$ is defined as
\[
\max_{f\in B_n}d(f,RM(r,n))=\max_{f\in B_n}\min_{\deg(g)\le r}d(f,g).
\]
For more information on the covering radius, we refer to \cite{Gao,Wang}. It is known that RM$(r,m)$ has minimum Hamming distance $2^{m-r}$. That is, the weight of an $m$-variable Boolean function of degree $r$ is at least $ 2^{m-r}$ \cite{Kasami1,MC}

 The {\em Walsh transform} of a given function $f \in B_n$ is the integer-valued function over $\mathbb{F}_{2}^n$
defined by
\[
W_f(\omega)=\sum_{x\in \mathbb{F}_{2}^n}(-1)^{f(x)+\omega \cdot x},
\]
where $\omega \in \mathbb{F}_{2}^n$, and $\omega \cdot x=\omega_1x_1+\cdots+\omega_nx_n$ is the usual inner product.

 It is easy to see that a Boolean function $f$ is  balanced if and only if $W_f(0)=0$.  If $W_f(\omega)=0$
for any $\omega\in \mathbb{F}_{2}^{n}$ satisfying $0\leq wt(\omega)\leq t$, then $f$ is called $t$-{\em resilient} \cite{Sun,Xiao}.

The nonlinearity of $f$ can be determined by
\[
nl(f)=2^{n-1}-\frac{1}{2}\max_{\omega \in \mathbb{F}_{2^n}}|W_f(\omega)|.
\]
By the Parseval formulae, $\max_{\omega}|W_f(\omega)|\ge 2^{n/2}$. Therefore,
\[
nl(f)\le 2^{n-1}-2^{\frac{n}{2}-1}.
\]
 If the upper bound is achieved, then $f$ is a {\em bent} function. Clearly, bent function only exists for $n$ even, and its algebraic degree is at most $\frac{n}{2}$ \cite{Car1,Car2,Rothaus}.

For a real function $f$ on $F_2^n$, the Fourier-Hadamard transform of $f$ is the integer-valued function over $\mathbb{F}_{2}^n$
defined by
\[
H_nf(\omega)=\sum_{x\in \mathbb{F}_{2}^n}f(x)(-1)^{\omega \cdot x}.
\]
Writing $(f_1,f_2)=\sum_{x}f_1(x)f_2(x)$, character orthogonality gives
\[
H_n^2=2^nI, \ \sum_{\omega\in F_2^n}(H_nf)(\omega)^2=2^n\sum_{\omega\in F_2^n}f(\omega)^2,
\]
and $(H_nf_1,f_2)=(f_1,H_nf_2)$ \cite{Car}.

Let $f_1,f_2\in B_n$. We use $||$ to denote the concatenation, that is,
\[
(f_1||f_2)(x_1,\ldots,x_n,x_{n+1})=(x_{n+1}\oplus 1)f_1(x_1,\ldots,x_n)\oplus x_{n+1}f_2(x_1,\ldots,x_n).
\]

 \section{The maximum nonlinearity of 8-variable balanced Boolean functions is 116}
We use $NLB_8$ to denote the maximum nonlinearity of an 8-variable balanced Boolean function. Assume that $NLB_8=118$. Then there exists a balanced $f\in B_8$ with $|W_f(\omega)|\le 20$ for any $\omega\in F_2^8$. If $\omega\neq 0$, then we have
\begin{equation}
 W_f(\omega)=-256+4|x\in supp(f): \ \omega \cdot x=1|.
\end{equation}
Let $b=\oplus_{x\in supp(f)} x$. Then
\[
|x\in supp(f) : \ \omega \cdot x=1|=\sum_{x\in supp(f)} \omega \cdot x\equiv\omega \cdot b\ (mod \ 2).
\]
Therefore, $W_f(\omega)\equiv 4\omega\cdot b \ (mod \ 8)$. If $b=0$, then $W_f(\omega)$ is a multiple of 8 and $|W_f(\omega)|\le 16$. That is, $f$ is a bent function, which is contradictory to the fact that $f$ is balanced. Hence $b\neq 0$. Choose an invertible matrix $M$ with $Me_8=b$ and replace $f(x)$ by $f(Mx)$, where $e_8=(0,...,0,1)$. Then the support sum becomes $e_8$ and the set of the Walsh spectrum stays the same. In the following, we take this new function as $f$, and we have
\[
b=\oplus_{x\in supp(f)} x=e_8.
\]
\subsection{A congruence connecting the Walsh spectrum and the quadratic form}
In this subsection, we will deduce a congruence connecting the Walsh spectrum and the quadratic form.
\begin{lemma} Let $t\in F_2^n$. Then
\[
\oplus_{i=1}^n t_i\equiv \sum_{i=1}^n t_i+2\sum_{1\le i<j\le n} t_it_j \ (mod \ 4).
\]
\end{lemma}
\proof Clearly,
\[
\oplus_{i=1}^n t_i=\frac{1-\prod_{i=1}^n(1-2t_i)}{2}.
\]
Since
\[
\prod_{i=1}^n(1-2t_i)=1+\sum_{\emptyset\neq I\subseteq\{1,...,n\}}(-2)^{|I|}\prod_{i\in I}t_i,
\]
we have
\[
\oplus_{i=1}^n t_i=\sum_{\emptyset\neq I\subseteq\{1,...,n\}}(-2)^{|I|-1}\prod_{i\in I}t_i,
\]
and the result follows.
\endproof

\begin{proposition} If $NLB_8=118$, then there exists a balanced $g\in B_8$ of nonlinearity 118 satisfying $\oplus_{x\in supp(g)} x=e_8$ and
\begin{equation}
W_g(a,c)\equiv 4c+8Q(a)\ (mod\ 16),
\end{equation}
for any $(a,c)\in F_2^7\times F_2$, where $Q(a)\in B_7$ is a quadratic function with $Q(0)=0$.
\end{proposition}
\proof For the function $f$ defined as above, let $X_i=\sum_{x\in supp(f)}x_i$ and $X_{ij}=\sum_{x\in supp(f)}x_ix_j$. By Lemma 1, for any $\omega, x\in F_2^8$, we have
\[
\omega\cdot x=\oplus_{i=1}^{8} \omega_ix_i\equiv \sum_{i=1}^8 \omega_ix_i+2\sum_{1\le i<j\le 8} \omega_ix_i\omega_jx_j\ (mod \ 4).
\]
Summing over $supp(f)$ and using (1), we have
\[
 W_f(\omega)\equiv 4\sum_{i=1}^8 X_i\omega_i+8\sum_{1\le i<j\le 8} X_{ij}\omega_i\omega_j\ (mod \ 16).
\]
Since the support sum is $e_8$, there exist integers $D_i$ such that $X_i=e_{i8}+2D_i$, where $e_{i8}\in F_2$ and it equals one only when $i=8$. Define the quadratic function
\[
Q_0(\omega)=\sum_{i=1}^8(D_i \ mod\ 2)\omega_i+\sum_{1\le i<j\le 8} (X_{ij}\ mod\ 2)\omega_i\omega_j.
\]
Then $Q_0(0)=0$ and
\[
W_f(a,c)\equiv 4c+8Q_0(a,c)\ (mod \ 16).
\]
Since $8Q_0\equiv 8(Q_0 (mod \ 2))\ (mod \ 16)$, $Q_0$ can be viewed as a Boolean function. Write
\[
Q_0(a,c)=Q(a)\oplus c(a\cdot s\oplus t),
 \]
 where $s\in F_2^7$, $t\in F_2$ and $Q\in B_7$. Clearly, $s$ and $t$ stay the same for any $(a,c)$, and $Q(0)=0$. Let $d=(s,t)$ and $g(x)=f(x\oplus d)$. Then $g\in B_8$ preserves the set of the Walsh coefficients and the support sum $b$ (since $\sum_{x\in supp(g)}x=\sum_{x\in supp(f)}x+128d\equiv b\ (mod \ 2)$). Clearly,
\[
W_g(a,c)=(-1)^{(a,c)\cdot d}W_f(a,c)=(-1)^{a\cdot s+ct}W_f(a,c).
\]
If $c=0$, then
\[
W_g(a,c)=\pm W_f(a,c)\equiv 8Q(a)\ (mod\ 16).
\]
If $c=1$ and $a\cdot s+t=0$, then
\[
W_g(a,c)=W_f(a,c)\equiv 4+8Q(a)\ (mod\ 16).
\]
If $c=1$ and $a\cdot s+t=1$, then
\[
W_g(a,c)=-W_f(a,c)\equiv -4+8(Q(a)+1)\equiv 4+8Q(a)\ (mod\ 16),
\]
and the result follows.
\endproof

By Proposition 1, we have the following Table 1 on the distribution of $W_g(a,c)$, where $a\in F_2^7$ and $c\in F_2$.
\begin{table}[!h]
\caption{Distribution of $W_g(a,c)$}
 \begin{center}\begin{tabular}{|c|c|c|c|}
  \hline
$Q(a)$ & $W_g(a,0)$ &  $W_g(a,1)$ \\
  \hline
  0 & -16, 0, 16  &  -12, 4, 20  \\
  \hline
  1 & -8, 8 &  -20, -4, 12  \\
  \hline
 \end{tabular}
     \end{center}
     \end{table}
\endproof
In the following, we use $h_i$, $0\le i\le 5$ to denote the number of Walsh coefficients of absolute value $4i$. Table 1 shows that all zero coefficients occur in $W_g(a,0)$, all coefficients of magnitude 4 occur in $W_g(a,1)$, and exactly $wt(Q)$ positions have magnitude 8 in $W_g(a,0)$. In
particular, $h_2 = wt(Q)$ and $h_0\ge 1$.

Assume that $NLB_8=118$. Then we obtain a quadratic form $Q$ (here quadratic forms may include linear terms). We will investigate $Q$ further and deduce a contradiction.

\subsection{The polar rank of the quadratic form $Q$ must be 6}
For a quadratic Boolean function $f$ with $f(0)=0$, its polar form is $B_f(x,y)=f(x+y)+f(x)+f(y)$. Its radical is the set of $y$ with $B_f(x,y)=0$ for all $x$. The polar rank means the rank of $B_f$, and nondegenerate means that its radical is zero. It is known that the polar rank must be even \cite{Car}. In this section, we will prove that the polar rank of $Q$ must be 6.

\begin{lemma} [Proposition 55 of \cite{Car}] Let $f\in B_n$ be quadratic with $f(0)=0$ and polar rank $2r$. Then $|W_f(\omega)|$ is supported on an affine subspace of size $2^{2r}$, and all its nonzero values have magnitude $2^{n-r}$. Therefore,
\[
wt(f)\in \{2^{n-1},2^{n-1}-2^{n-r-1},2^{n-1}+2^{n-r-1}\}.
\]
If $r=0$, then $wt(f)\in \{0, 2^{n-1}\}$.
\end{lemma}

Write $g=g_0||g_1$, where $g_0,g_1\in B_7$. Let
\[
u(y)=\frac{(-1)^{g_0(y)}+(-1)^{g_1(y)}}{2},\ v(y)=\frac{(-1)^{g_0(y)}-(-1)^{g_1(y)}}{2}.
\]
Then $u,v\in \{0,\pm1\}$, $u^2+v^2=1$, $uv=0$, and
\begin{equation}
W_g(a,0)=2H_7u(a), \ W_g(a,1)=2H_7v(a).
\end{equation}

\begin{lemma} We have
\begin{equation}
(W_g(a,1), (-1)^{Q(a)})=256+32(h_0+h_1)+24wt(Q).
\end{equation}
\end{lemma}
\proof By Table 1, we have
\begin{eqnarray*}
\sum_{a\in F_2^7} W_g(a,0)^2&=&0\cdot h_0+64wt(Q)+256(128-h_0-wt(Q))\\
&=& 32768-256h_0-192wt(Q).
\end{eqnarray*}
Then by the Parseval formulae, we have
\begin{equation}
\sum_{a\in F_2^7} W_g(a,1)^2=32768+256h_0+192wt(Q).
\end{equation}
By Table 1, $|W_g(a,1)-4(-1)^{Q(a)}|\in \{0,16\}$, and it equals 0 if and only if $|W_g(a,1)|=4$. Therefore,
\[
\sum_{a\in F_2^7} (W_g(a,1)-4(-1)^{Q(a)})^2=256(128-h_1).
\]
Then by (5), we have
\[
256h_0+192wt(Q)-8(W_g(a,1), (-1)^Q(a))+2048=-256h_1,
\]
and the result follows.
\endproof

\begin{proposition} The polar rank of $Q$ is 6.
\end{proposition}
\proof Let the polar rank of $Q$ be $2r$. By (3), we have
\[
(W_g(a,1), (-1)^{Q(a)})=2(H_7v,(-1)^Q)=2(v,H_7(-1)^Q).
\]
Then by Lemma 2,
\[
(W_g(a,1), (-1)^{Q(a)})\le 2\sum_{y}|(H_7(-1)^Q)(y)|=2\cdot2^{2r}\cdot 2^{7-r}=2^{8+r}.
\]
Therefore, by Lemma 3, we have
\begin{equation}
256+32(h_0+h_1)+24wt(Q)\le 2^{8+r}.
\end{equation}
If $r=0$, then $256+32h_0\le 256$, which is contradictory to $h_0\ge 1$. If $r=1$, by Lemma 2, $wt(Q)\ge 32$. Then by (6), we have $256+24\cdot 32\le 512$, which is a contradiction. If $r=2$, then by Lemma 2 and (6), we have $wt(Q)\ge 48$ and $256+24\cdot 48\le 1024$, again a contradiction. Therefore, $r=3$, and the result follows.
\endproof

By Proposition 2, the polar rank of $Q$ must be 6. Then by Lemma 2, $wt(Q)\in\{56,64,72\}$. Moreover, by (6),
\begin{equation}
h_0+h_1\le 56-\frac{3}{4}wt(Q).
\end{equation}
Since
\[
32768-256h_0-192wt(Q)=\sum_{a\in F_2^7}W_g(a,0)^2=4\sum_{a\in F_2^7}(H_7u)(a)^2=4\cdot 128|supp(u)|,
\]
we have
\begin{equation}
|supp(u)|=64-\frac{h_0}{2}-\frac{3wt(Q)}{8}.
\end{equation}
Since $|supp(u)|$ is an integer and $h_0\ge 1$, we have $h_0\ge 2$ must be even.

\subsection{$wt(Q)$ must be 56}

We now  split the seven-dimensional space as $F_2^6\times F_2$. Since the polar rank of $Q$ is 6, its radical has dimension 1.  Choose its nonzero vector $r_0$ as the last vector of a basis. Then $F_2^7=V_0\oplus <r_0> $, where $V_0$  is  a six-dimensional complement of the radical. Clearly, for any $a\in F_2^6$, we have $B_Q((a,0),(0,1))=0$.

Let $p(a)=Q(a,0)$ and $\lambda=Q(0,1)$. Then
\begin{eqnarray*}
Q(a,j)&=&Q((a,0)+j(0,1))\\
&=& Q(a,0)+Q(j(0,1))+B_Q((a,0),j(0,1))\\
&=&p(a)+\lambda j.
\end{eqnarray*}
Here  $p$  is the restriction to $V_0$, and $p(0)=0$.  Its polar
form is nondegenerate: a vector in $V_0$ orthogonal to all of $V_0$ is also orthogonal to the radical, hence
to the whole space. It belongs to both $V_0$ and the radical, so is zero.

By Lemma 2, $p\in B_6$ is a bent function. Let $\epsilon=H_6((-1)^p)/8$.  Then $\epsilon\in\{\pm 1\}$ and $H_6 \epsilon=8(-1)^p $. Moreover,
\[
(H_7(-1)^Q)(y,t)=(H_6(-1)^p)(y)\sum_{j\in F_2}(-1)^{j(\lambda+t)}=\left \{ \begin{array} {l} 16\epsilon(y), \  \ t=\lambda \\
  0, \ \ \ \ \ \ \ \ \ t\neq \lambda . \\
\end{array} \right.
\]
Write $L=\{(y,\lambda):y\in F_2^6\}$ and $R=F_2^7\setminus L$. Both halves are identified with $F_2^6$ by the coordinate $y$.  In particular, $v_L(y)=v(y,\lambda)$ and $v_R(y)=v(y,1+\lambda)$; the restrictions $u_L$ and $u_R$ are defined in the
same way.

In the following, we will deduce that $v_L(y)=\pm 1$ and the associated 6-variable Boolean function is of nonlinearity greater than 26 whose Hamming distance to a bent function is between 4 and 6.

On $L$, let $z$  count the positions with $v_L=0$ and let $w$  count those with $ v_L=-\epsilon$. Then
\[
(v_L,\epsilon)=-w+(64-z-w)=64-z-2w.
\]
Furthermore,
\[
(W_g(a,1), (-1)^{Q(a)})=2(v,H_7(-1)^Q)=32(v_L,\epsilon)=32(64-z-2w).
\]
Then by (4), we have
\begin{equation}
z+2w=56-h_0-h_1-\frac{3}{4}wt(Q).
\end{equation}
On $R$, we have
\[
|supp(v_R)|=|supp(v)|-|supp(v_L)|=(128-|supp(u)|)-(64-z).
\]
Then by (8), we have
\begin{equation}
|supp(v_R)|=\frac{h_0}{2}+\frac{3wt(Q)}{8}+z.
\end{equation}

\begin{lemma} Let $z$ and $w$ be defined as above. Then
\begin{equation}
|supp(v_R)|\le 4+\frac{3}{2}z+5w.
\end{equation}
\end{lemma}
\proof Clearly, for any $(a,j)\in F_2^7$, we have
\[
H_7v(a,j)=(-1)^{j\lambda}(H_6v_L(a)+(-1)^jH_6v_R(a)).
\]
Since $|2H_7v|=|W_g(a,1)|\le 20$, we have
\begin{equation}
|H_6v_L(a)|+|H_6v_R(a)|\le 10.
\end{equation}
Furhtermore,
\begin{equation}
\sum_a|H_6v_L(a)|\ge \sum_a(-1)^{p(a)}H_6v_L(a)=8(\epsilon,v_L)=8(64-z-2w).
\end{equation}
By the Parseval formulae, we have
\begin{equation}
\sum_a(H_6v_L(a))^2=64(64-z), \  \sum_a(H_6v_R(a))^2=64|supp(v_R)|.
\end{equation}
Therefore, by (12), (13) and (14), we have
\begin{eqnarray*}
64|supp(v_R)|&=&\sum_a(H_6v_R(a))^2\le\sum_a(10-|(H_6v_L(a)|)^2\\
&\le& 6400-160(64-z-2w)+64(64-z)\\
&=& 256+160(z+2w)-64z,
\end{eqnarray*}
and the result follows.
\endproof

\begin{proposition} We have $wt(Q)=56$.
\end{proposition}
\proof By Lemma 2,  $wt(Q)\in \{56,64,72\}$. If $wt(Q)=72$, then by (7), $h_0+h_1\le 2$. Therefore, $h_0=2$ and $h_1=0$. Thus (9) and (10) give $z+2w=0$ and $|supp(v_R)|=28+z$. Therefore, $z=w=0$ and $|supp(v_R)|=28$, which is contradictory to (11). If $wt(Q)=64$, then by (9) and $h_0\ge 2$, we have
\[
z+2w=56-h_0-h_1-48\le 6.
\]
Then by Lemma 4, $|supp(v_R)|\le 4+\frac{5}{2}(z+2w)-z\le 19-z$. Furthermore, by (10),
\[
|supp(v_R)|=\frac{h_0}{2}+24+z\ge 25+z.
\]
 Therefore, $19-z\ge |supp(v_R)|\ge 25+z$, which is a contradiction, and the result follows.
\endproof
Since $Q(a,j)=p(a)+\lambda j$, if $\lambda=1$, then $Q(a,0)$ and $Q(a,1)$ contains exactly one value equal to one, giving $wt(Q)=64$. Therefore, $\lambda=0$ and $wt(p)=wt(Q)/2=28$. In this remaining case, by (9) and (10), we have
\begin{equation}
z+2w=14-h_0-h_1\le 12, \ |supp(v_R)|=21+\frac{h_0}{2}+z\ge 22+z.
\end{equation}

 \subsection{$|v_L|=1$}
Since $u^2+v^2=1$, $supp(u_L)=z$. Set $A(\omega)=W_g(\omega,0)$, where $\omega\in F_2^7$. By (3), $A=2H_7u$. Therefore, for $a\in F_2^6$, we have
\[
\frac{A(a,0)}{2}=H_6u_L(a)+H_6u_R(a),\ \frac{A(a,1)}{2}=H_6u_L(a)-H_6u_R(a).
\]
Their sum is $2H_6u_L(a)$. By (2) and $Q(a,j)=p(a)$,
\[
\frac{A(a,0)}{2}\equiv \frac{A(a,1)}{2}\equiv4p(a) \ (mod\ 8).
\]
Therefore, $2H_6u_L(a)\equiv8p(a)\equiv0\ (mod\ 8)$ and
\begin{equation}
H_6u_L(a)\equiv0\ (mod\ 4).
\end{equation}
\begin{lemma} If $z\in\{4,6\}$, then
\[
|\{a\in F_2^6: p(a)=0, H_6u_L(a) =\pm4\}|\ge 6.
\]
\end{lemma}
\proof  Since $p$ is a bent function and $W_p(0)=8$, we have
\begin{eqnarray*}
&&\sum_{p(a)=0}(H_6u_L(a))^2-\sum_{p(a)=1}(H_6u_L(a))^2\\
&=& \sum_{a}(-1)^{p(a)}(H_6u_L(a))^2\\
&=& \sum_{x,y}u_L(x)u_L(y)\sum_a(-1)^{p(a)+a\cdot(x+y)}\\
&=& \sum_{x=y\in supp(u_L)}W_p(0)+\sum_{x\neq y\in supp(u_L)}u_L(x)u_L(y)W_p(x+y)\\
&\ge& 8z-8z(z-1))=16z-8z^2.
\end{eqnarray*}
Moreover, by the Parseval formulae, we have
\[
\sum_{p(a)=0}(H_6u_L(a))^2+\sum_{p(a)=1}(H_6u_L(a))^2=64z.
\]
Therefore,
\[
2\sum_{p(a)=0}(H_6u_L(a))^2\ge 16z-8z^2+64z=192.
\]
Since (16) and $|H_6u_L(a)|\le z\le 6$, we have $H_6u_L(a)\in \{0,\pm4\}$.
Hence,
\[
|\{a\in F_2^6: p(a)=0, H_6u_L(a) =\pm4\}|\ge 192\div2\div16=6.
\]
\endproof
\begin{proposition} $z$ must be zero.
\end{proposition}
\proof Clearly,
\[
H_6u_L(a)=\sum_{x\in supp(u_L)}u_L(-1)^{a\cdot x}\equiv |supp(u_L)|=z\ (mod\ 2).
\]
Then by (16), $z$ is even. If $z=2$, let $supp(u_L)=\{x_1,x_2\}$. Choose $a$ such that $(-1)^{a\cdot (x_1+x_2)}=u_L(x_1)u_L(x_2)$. Then $u_L(x_1)(-1)^{a\cdot x_1}=u_L(x_2)(-1)^{a\cdot x_2}$, and $|H_6u_L(a)|=2$, which is contradictory to (16).
  If $z\ge 8$, then by (11) and (15), we have
\[
|supp(v_R)|\le 4+\frac{5}{2}(z+2w)-z\le4+\frac{5}{2}\times 12-z\le 26,
\]
and $|supp(v_R)|\ge 22+8=30$,
which is a contradiction. Let $S=\{a\in F_2^6: p(a)=0, H_6u_L(a) =\pm4\}$. If $z\in\{4,6\}$, by Lemma 5, $|S|\ge 6$. If $a\in S$, then $Q(a)=p(a)=0$. Hence, by Table 1,
\[
\{A(a,0),A(a,1): a\in F_2^6\}=\{W_g(\omega,0): \omega\in F_2^7\}\subseteq\{0,\pm16\}.
\]
Furthermore,
\[
\frac{A(a,0)}{2}+\frac{A(a,1)}{2}=2H_6u_L(a)=\pm8.
\]
 Therefore, exactly one of $A(a,0)$ and $A(a,1)$ is zero. Hence, $h_0\ge |S|\ge6$.
By (15), we have
\[
|supp(v_R)|=21+\frac{h_0}{2}+z\ge 24+z.
\]
However, by (11) and (15), we have
\[
|supp(v_R)|\le 4+\frac{5}{2}(z+2w)-z=4+\frac{5}{2}(14-h_0-h_1)-z\le 24-z.
\]
This is a contradiction, and the result follows.
\endproof
By Proposition 4, we have $|v_L|=1$.

 \subsection{Completion of the proof}

\begin{lemma} Let $f_1,f_2\in B_6$, where $f_1$ is a bent function. If $d(f_1,f_2)\in\{4,5,6\}$, then $nl(f_2)\le26$.
\end{lemma}
\proof Let $h=f_1\oplus f_2$. Since $0<wt(h)<8$, the minimum distance of $RM(3,6)$ gives $\deg(h)\ge4$. As $\deg(f_1)\le3$, we have $\deg(f_2)\ge4$, so $nl(f_2)<28$. If $wt(h)\in\{4,6\}$, then $nl(f_2)$ is even, and the result follows.

Suppose $wt(h)=5$, and write $supp(h)=\{x_1,\ldots,x_5\}$. Let $f_0$ be the dual bent function of $f_1$. Define
\[
\phi_i(\omega)=f_0(\omega)\oplus f_1(x_i)\oplus\omega\cdot x_i,
\quad s_i(\omega)=(-1)^{\phi_i(\omega)},
\]
and $N(\omega)=|\{i:\phi_i(\omega)=1\}|$. All sums below are over $\omega\in F_2^6$. By duality and character orthogonality,
\[
\sum_\omega s_i=8,\qquad
\sum_\omega s_is_j=0\quad(i\ne j).
\]
For distinct $i,j,k$, the bent property of $f_0$ also gives
\[
\sum_\omega s_is_js_k\in\{\pm8\},\qquad
\sum_\omega\prod_{i=1}^5s_i\in\{\pm8\},
\]
since each odd product is $(-1)^{f_0(\omega)}$ times an affine character. Expanding products of $(1-s_i)/2$, we obtain
\[
\sum_\omega N(\omega)=140,\qquad
\sum_\omega{N(\omega)\choose2}=120,
\]
and
\[
\sum_\omega{N(\omega)\choose3}
=\sum_{i<j<k}\frac{64-24-\sum_\omega s_is_js_k}{8}\ge40.
\]
If $N(\omega)\le3$ everywhere, put $n_r=|\{\omega:N(\omega)=r\}|$. Then $n_2+3n_3=120$ and $n_3\ge40$, forcing $n_3=40$, $n_2=0$. The first moment and $\sum_{r=0}^3n_r=64$ give $n_1=20$, $n_0=4$. Hence
\[
\sum_\omega\prod_{i=1}^5s_i
=\sum_\omega(-1)^{N(\omega)}=4-20-40=-56,
\]
a contradiction. Thus $N(\omega_0)\ge4$ for some $\omega_0$. Taking $l(x)=\omega_0\cdot x\oplus f_0(\omega_0)$, duality gives $d(f_1,l)=28$, and therefore
\[
nl(f_2)\le d(f_2,l)=28-N(\omega_0)+(5-N(\omega_0))=33-2N(\omega_0)\le25.
\]
This completes the proof.
\endproof

Since $\epsilon=H_6((-1)^p)/8$, we have  $|\epsilon|=1$ and $H_6 \epsilon=8(-1)^p $. Let $\epsilon=(-1)^{f_1}$, where $f_1\in B_6$. Then $|W_{f_1}(\omega)|=8$, and $f_1$ is a bent function. We now can complete the proof.
\begin{theorem} The maximum nonlinearity of 8-variable balanced Boolean functions is 116.
\end{theorem}
\proof Assume $NLB_8=118$. Then according to the above discussion, $|v_L|=1$ and by (9), $v_L$ is obtained from $\epsilon$ by changing
\[
w=\frac{1}{2}(56-h_0-h_1-42)\le \frac{1}{2}(56-2-42)=6
\]
 signs. Furthermore, by (11) and (15), we have $22\le |supp(v_R)|\le 4+5w$. Hence, $w\ge 4$. That is,
\[
w\in\{4,5,6\}.
\]
Let $v_L=(-1)^{f_2}$, where $f_2\in B_6$. Then $d(f_1,f_2)\in\{4,5,6\}$. Moreover, by (12), we have $|W_{f_2}(a)|\le 10$, for any $a\in F_2^6$. Therefore, $nl(f_2)\ge 27$, which is contradictory to Lemma 6,
and the result follows.
\endproof

\section{Conclusion and open problems}
In this paper, we proved that the maximum nonlinearity of balanced Boolean functions in 8 variables is 116. A similar analysis for $\ge 10$ variables may deduce more interesting results. We leave this for future work. Finally, we list some related open problems that could be further considered.

\begin{open problem} Does there exist a 9-variable balanced Boolean function with nonlinearity 242?
\end{open problem}

\begin{open problem} Determine the covering radius of RM(1,9).
\end{open problem}

\begin{open problem} Does there exist a 10-variable balanced Boolean function with nonlinearity 494?
\end{open problem}

\begin{open problem} Let $f\in B_{2n}$ be balanced. Do we have $nl(f)\le 2^{2n-1}-2^{n-1}-4$ for any $n\ge 4$?
\end{open problem}

\section*{Acknowledgment}

The authors would like to thank the financial support from the National Natural Science Foundation of China (Grant 62172230). We also thank Dr. William Millan for pointing out an error in our original proof of Lemma 6.

\end{document}